\documentclass[10pt,aps,pra,twocolumn,floatfix,nofootinbib,superscriptaddress,longbibliography]{revtex4-2}
\usepackage{placeins}

\usepackage{bm,graphicx,mathrsfs,amsmath,amssymb,mathtools,makecell,bbm, amsthm,dsfont,color,nicefrac,framed,enumitem,tikz,physics,wrapfig,amsfonts,tcolorbox,times,txfonts,lipsum,nicematrix,fontspec}
\usepackage{float}

\newfontfamily\magic[Path=./,Ligatures=TeX,Scale=MatchLowercase,ItalicFont=MagicMedieval.ttf]{MagicMedieval.ttf}

\usepackage[colorlinks=true,linkcolor=equationcolor,citecolor=equationcolor]{hyperref}

\hypersetup{urlcolor=refcolor}
\definecolor{equationcolor}{RGB}{108,153,224}
\definecolor{refcolor}{RGB}{214,86,86}
\definecolor{changescolor}{RGB}{188, 104, 104}

\makeatletter
\def\blfootnote{\gdef\@thefnmark{}\@footnotetext}
\makeatother

\newcommand{\stab}{\operatorname{stab}}

\renewcommand{\v}[1]{\ensuremath{\boldsymbol #1}}
\newcommand{\ms}[1]{\textsf{#1}}

\newcommand{\iden}{\mathbbm{1}}

\newcommand{\E}[1]{\mathcal{E}}

\def\E{ {\cal E} }

\newenvironment{sproof}{%
  \proof}{\endproof}
  
\newcommand{\mtg}[1]{\ensuremath{\mathord{\text{\magic #1}}}}
\newcommand{\M}{\mtg{M}}

\newtheorem{thm}{Theorem}
\newtheorem{res}{Result}

\newtheorem{lem}[thm]{Lemma}

\newtheorem{cor}[thm]{Corollary}

\usepackage[normalem]{ulem}

\begin{document}

\title{Trading athermality for nonstabiliserness}

\author{A. de Oliveira Junior}
\email{alexssandredeoliveira@gmail.com}
	\affiliation{Center for Macroscopic Quantum States bigQ, Department of Physics,
Technical University of Denmark, Fysikvej 307, 2800 Kgs. Lyngby, Denmark}
\author{Rafael A. Macêdo}
    \affiliation{Physics Department, Federal University of Rio Grande do Norte, Natal, 59072-970, Rio Grande do Norte, Brazil}
    \affiliation{International Institute of Physics, Federal University of Rio Grande do Norte, 59078-970, Natal, RN, Brazil}
\author{Jakub Czartowski}
	\affiliation{School of Physical and Mathematical Sciences, Nanyang Technological University, 21 Nanyang Link, 637361 Singapore, Republic of Singapore}
    \affiliation{School of Physics, Trinity College Dublin, Dublin 2, Ireland}
\author{Jonatan Bohr Brask}
    \affiliation{Center for Macroscopic Quantum States bigQ, Department of Physics,
Technical University of Denmark, Fysikvej 307, 2800 Kgs. Lyngby, Denmark}
\author{Rafael Chaves}
    \affiliation{International Institute of Physics, Federal University of Rio Grande do Norte, 59078-970, Natal, RN, Brazil}
    \affiliation{School of Science and Technology, Federal University of Rio Grande do Norte, Natal, Brazil}
\date{\today}

\begin{abstract}
Quantum advantage arises from quantum states that cannot be efficiently simulated on a classical computer. Such states are characterised by a property known as nonstabiliserness. In this work, we investigate whether nonstabiliserness can be generated by placing an initially stabiliser state in contact with a heat bath. Under minimal thermodynamic assumptions, we derive a necessary and sufficient condition for when this is possible. This yields an analytic characterisation of all nonstabiliser qubit states reachable through such thermal processes, together with explicit bounds on their nonstabiliserness. This, in turn, allows us to identify optimal regimes for generating this resource, including the Hamiltonians that maximise nonstabiliserness and the critical temperatures at which it emerges. Beyond the qubit case, we establish a general trade-off between the nonstabiliserness attainable under thermal operations and the initial nonequilibrium free energy of the system.
\end{abstract}

\maketitle

\section{Introduction}
A fundamental problem in quantum information is to characterise the boundary between classical and quantum computation. Although entanglement was once believed to be the essential ingredient for quantum advantage, the Gottesman–Knill theorem showed that a specific class of states---even when highly entangled---remain efficiently simulable on a classical computer~\cite{gottesman1997stabilizercodesquantumerror,Aharonov1997,gottesman1998,Aaronson2004,Maarten2013}. Even so, stabiliser states are indispensable for quantum error correction and fault-tolerant quantum computation~\cite{Shor1996,Gottesman1998f,Kitaev2003,Raussendorf2007}. This dichotomy led to the development of a resource theory of nonstabiliserness, which captures the extra ingredient necessary for universal quantum computation~\cite{Veitch2014, Howard2017,Chitambar2019}.

Such a framework poses two nontrivial problems. First, not every nonstabiliser state is useful for quantum computation~\cite{Campbell2010}. In practice, we prepare only noisy versions of target states. Whether these are sufficient is determined by distillability: can they be purified into useful magic states using only stabiliser operations? Second, preparing useful magic states is intrinsically difficult. For logical qubits in fault-tolerant settings, noise from the unavoidable coupling to a heat bath leads to thermalisation. This drives states toward equilibrium, which, at high temperatures, becomes stabiliser~\cite{bakshi2024high}. Yet, not every thermalisation process itself needs to be a stabiliser operation~\cite{Mukhopadhyay2018,Koukoulekidis2022,trigueros2025nonstabilizernesserrorresiliencenoisy,macedo2026littlethingheatdoes}. We are therefore led to ask: is the environment only an enemy that washes out nonstabiliserness, or can it be used to generate and manipulate it?

\begin{figure}[t]
    \vspace{1cm}
    \centering
    \includegraphics{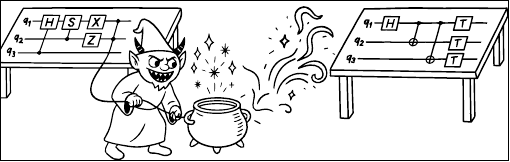}
    \caption{\emph{Thermal magician}. This paper shows that a heat bath is not just a source of noise, it can also brew nonstabiliserness from stabiliser states. We determine exactly when thermal processes create nonstabiliserness, how much they produce, and which temperatures and Hamiltonians are optimal.   
    }
    \label{F:thermal-magician}
\end{figure}
\begin{figure*}[t]
    \centering
    \includegraphics{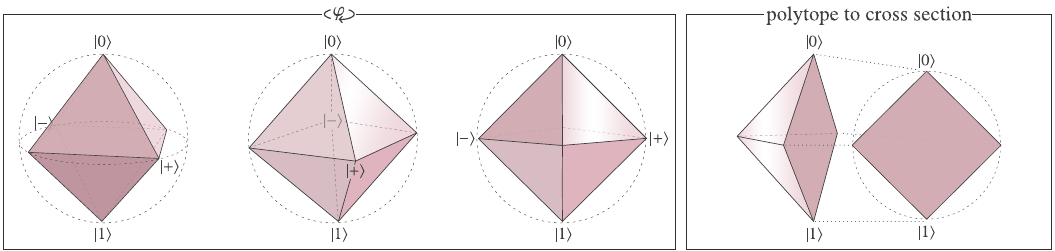}
    \caption{\emph{Stabiliser polytope}. The stabiliser polytope under rotations by an angle $\varphi$. Each panel displays a different orientation of the polytope, highlighting its octahedral symmetry. The right panel shows a cross-section of the Bloch ball together with the projection of the stabiliser polytope onto this cross-section}
    \label{F-stabiliser-polytope}
\end{figure*}

We answer this question by showing that the environment is not only an enemy: athermality---the property of being out of equilibrium with a reference thermal bath---can be traded for nonstabiliserness. First, we derive a general thermodynamic bound that applies to arbitrary finite-dimensional systems. This bound provides a universal necessary condition: the nonstabiliserness of any state reachable under thermal operations is limited by the system's initial nonequilibrium free energy. While not a full characterisation of reachability, it identifies athermality as the thermodynamic resource for generating nonstabiliserness. Second, for single qubits, we derive a necessary and sufficient condition for magic states to emerge from an initial stabiliser state under thermalisation processes with minimal assumptions. This framework allows us to characterise the full set of thermally reachable magic states, quantify how much nonstabiliserness can be produced, and identify operational thresholds such as the critical inverse temperature at which nonstabiliserness appears. The result is particularly relevant for logical magic state preparation in fault-tolerant quantum computation.

Our results show that the price to create nonstabiliserness is paid by bringing the state closer to thermal equilibrium, namely trading athermality for nonstabiliserness. This complements resource‑theoretic narratives for entanglement~\cite{deOliveiraJunior2025athermality}, coherence, and purity~\cite{Bartlett2007,Lostaglio2015descr,Lostaglio2015,Horodecki2015,Marvian2020,Takagi2020,Takagi2022}.

This paper is organised as follows. In Sec.~\ref{Sec:setting} we introduce the stabiliser framework for qubits and review the structure of thermal operations. In Sec.~\ref{Sec:main} we present our main result, providing a necessary and sufficient condition for the generation of nonstabiliserness under thermal operations, together with a general thermodynamic bound linking nonstabiliserness to athermality. In Sec.~\ref{Sec:optimal} we analyse optimal regimes, identifying the roles of temperature, coherence, and Hamiltonian orientation, and derive critical thresholds for nonstabiliserness generation and distillability. Finally, in Sec.~\ref{Sec:discussion} we discuss implications, extensions, and open questions, with technical details and proofs deferred to the Appendices.

\section{Setting the scene}\label{Sec:setting}

Given a qubit state $\rho$, what characterises it as nonstabiliser? We define a state to be nonstabiliser if it cannot be expressed as a convex mixture of stabiliser states~\cite{Veitch2012,Heinrich2019}. For a single qubit, the stabiliser states are the six eigenstates of the Pauli operators $\v\sigma=(X,Y,Z)$, corresponding to the six poles along the coordinate axes of the Bloch sphere. Their convex hull forms an octahedron known as the stabiliser polytope~[see Fig.~\ref{F-stabiliser-polytope}]. A general qubit state $\rho = \tfrac12(\iden + \v{r}\cdot\v{\sigma})$, with Bloch vector $\v{r}=(r_x,r_y,r_z)$, is stabiliser if and only if it lies inside this octahedron, a condition captured by
\begin{equation}\label{Eq:stabiliser-condition}
\norm{\v{r}}_1 =  \abs{r_x} + \abs{r_y} + \abs{r_z} \leq 1.
\end{equation}
Any state violating this inequality possesses nonstabiliserness.

Beyond a single qubit, the same geometric viewpoint still applies, but it comes with a combinatorial price. For $n$ qubits, the stabiliser polytope, denoted by $\operatorname{stab}_n$, is the convex hull of all pure $n$-qubit stabiliser states. However, its complexity increases dramatically with system size. In particular, the number of pure stabiliser states, equivalently the number of extremal points $\operatorname{ext}(\textrm{stab}_n)$ of the polytope, is
\begin{equation}
    |\operatorname{ext}(\operatorname{stab}_n)| = 2^n \prod_{k=1}^n (2^k+1) = 2^{\Theta(n^2)},
\end{equation}
where $\Theta(n^2)$ indicates quadratic scaling in the exponent. Thus, membership and optimisation over the full polytope quickly become intractable. This complexity is one reason why an exact geometric characterisation is especially valuable in the qubit case, where the stabiliser set reduces to the Bloch-ball octahedron~\eqref{Eq:stabiliser-condition}. At the same time, it also motivates the introduction of quantitative measures of nonstabiliserness that capture the distance of a state from the stabiliser polytope. We quantify deviation from the stabiliser polytope via the geometric monotone given by the minimal trace distance to the polytope, which we refer to as the nonstabiliserness of $\rho$:
\begin{equation}\label{Eq:nonstabiliserness-monotone}
\mathcal N\mathcal S(\rho) = \min_{\sigma \in \stab_n} T(\rho, \sigma),
\end{equation}
where $T(\rho, \sigma ) = \tfrac12| \rho-\sigma|_1$ is the trace distance~\cite{Nielsen2010}.

The unitaries that map Pauli operators to Pauli operators are known as Clifford unitaries. They preserve the stabiliser set, and therefore cannot map a stabiliser state to a nonstabiliser one (for a single qubit, they correspond to the rotational symmetries of the octahedron and merely permute its vertices)~(see Appendix~\hyperref[App:stabiliser-theory]{A-2} for a detailed discussion). Can general thermal processes---operations implemented using only a thermal bath---achieve what Clifford operations cannot?

We address this question using the framework of thermal operations~\cite{Janzing2000,horodecki2013fundamental,Brandao2013,Gour2018}. The system, described by a Hamiltonian $H$, is initially in a state $\rho$. Our focus is on qubit states\footnote{This is justified because, beyond the qubit case, constructing the reachable set under thermal operations remains an open problem. Although arbitrary state transformations under thermal operations are known~\cite{Gour2018}, in principle, to obey a complete set of necessary and sufficient conditions, these conditions are highly involved and difficult to characterise explicitly. Thus, beyond an implicit formulation, the problem remains open.} characterised by the ground-state population $p$ as the diagonal entry of $\rho$, and the coherence $c$ as the magnitude of the off-diagonal element in the energy eigenbasis. The bath, with Hamiltonian $H_{\ms{E}}$, is prepared in a Gibbs state at some inverse temperature $\beta = (k_B T)^{-1}$, where $k_B$ is the Boltzmann constant. The composite system and bath undergoes an energy‑preserving unitary $U$, satisfying $[U,H+H_{\ms{E}}]=0$, inducing a map
\begin{equation}\label{Eq:thermal-operations}
    \mathcal{E}(\rho):=\tr_{\ms E}\qty{U\qty[\rho\otimes \frac{e^{-\beta H_{\ms E}}}{\tr(e^{-\beta H_{\ms E}})}]U^{\dagger}}.
\end{equation}
States reachable from $\rho$ under~\eqref{Eq:thermal-operations}---thermally reachable set---are constrained jointly by energy conservation and time-translation symmetry~\cite{Lostaglio2015,Horodecki2015}. For a qubit, the accessible set forms a convex region $\mathcal{T}_+(\rho)$ that is rotationally symmetric about the energy axis~\cite{deOliveirajunior2022,vomEnde2022}. This axial symmetry reflects that only the magnitude of coherence matters under time-translation symmetry. The boundary encodes a trade-off: for any given target populations, only a limited amount of coherence can be retained. Two structural ingredients determine $\mathcal{T}_+(\rho)$. First, the energy populations evolve under Gibbs-preserving stochastic maps and are constrained by thermomajorisation~\cite{horodecki2013fundamental}. Writing $p$ and $ q$ for the initial and target ground-state population, the reachable $q$ lies in the interval
\begin{equation}
    I_\beta(p):=\qty[\min\{p, q^\star\},\,\max\{p, q^\star\}],
\end{equation}
where $q^\star:=1-\tfrac{1-\gamma}{\gamma}p$ and $\gamma$ is the Gibbs ground-state population of the system. Second, time-translation symmetry bounds the coherence independently of its phase. Let $c$ and $r$ be the initial and target coherence, the achievable $r$ for a target with ground-state population $q\in I_\beta(p)$ obeys the tight bound $r\leq \v{c}_{\max}(q)$, where~\cite{Horodecki2015,Lostaglio2015}:
\begin{equation}
   \!\!\!\v{c}_{\max}(q)\!:=\!|c|
\frac{\sqrt{[q(1-\gamma)\!-\!\gamma(1-p)][p(1-\gamma)\!-\!\gamma(1-q)]}}{|p-\gamma|}.
\end{equation}
\begin{figure}
    \centering
    \includegraphics{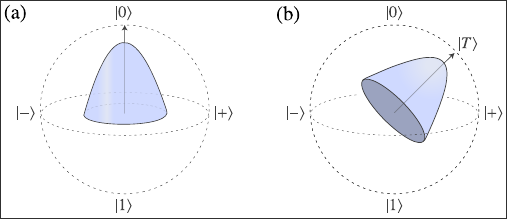}
    \caption{\emph{Reachable sets under thermal operations for different Hamiltonian orientations}. 
    For a fixed bath inverse temperature $\beta = 2$, the three panels show the sets of qubit states reachable under thermal operations, $\mathcal{T}_+(\rho)$, for Hamiltonians whose Bloch vectors point along (a) $(0,0,1)$, (b) $(\tfrac{1}{\sqrt{3}},\tfrac{1}{\sqrt{3}},\tfrac{1}{\sqrt{3}})$, respectively. We also highlight the magic state $\ket{\ms{T}}$.}
    \label{F-future-thermal-cone}
\end{figure}
For qubits, these constraints are jointly necessary and sufficient: for every $q\in I_\beta(p)$ the boundary value $r=c_{\max}(q)$ is attainable by a thermal operation, and by convexity of thermal operations all intermediate $0\le r\le\v{c}_{\max}(q)$ are attainable as well. See Appendix~\hyperref[App:coherent-thermal-cones]{A-2} for a detailed discussion. The interplay between the population and coherence constraints (Lemmas~\ref{Lem:population-evolution} and~\ref{Lem:coherence-evolution}) is illustrated in Fig.~\ref{F-future-thermal-cone}, which depicts the qubit future thermal cone $\mathcal{T}_+(\rho)$.

The reachability of states under thermal processes depends crucially on how far the system is from equilibrium. This motivates the introduction of measures of athermality, among which the nonequilibrium free energy plays a central role. Consider a $d$-dimensional quantum system described by a Hamiltonian $H$, prepared in an out-of-equilibrium state $\rho$, with corresponding Gibbs state $\tau_\beta \propto e^{-\beta H}$. The deviation from thermal equilibrium can be quantified by the nonequilibrium free-energy difference
\begin{equation}
\Delta F_\beta(\rho) := F_\beta(\rho) - F_\beta(\tau_\beta).
\end{equation}
Here $F_\beta(\rho) := \tr(\rho H) - \beta^{-1} S(\rho)$ is the nonequilibrium free energy, with $S(\rho) := -\tr(\rho \log \rho)$ the von Neumann entropy, while $F_\beta(\tau_\beta) = -\beta^{-1} \log Z$ is the equilibrium free energy, with $Z = \tr(e^{-\beta H})$ the partition function.

\section{Main results}\label{Sec:main}

Our results come in two layers. First, we identify a general thermodynamic constraint, valid in arbitrary finite dimension, that limits how much nonstabiliserness can be generated under thermal operations from a given input. Second, we focus on single qubits and derive a necessary-and-sufficient criterion for nonstabiliserness generation from stabiliser states under thermal operations.

Our first result is a free-energy bound that is independent of the microscopic implementation, control protocol, or system size. More precisely:

\begin{res}[Athermality \& nonstabiliserness]\label{Res:tradeoff}
The nonstabiliserness of any state $\sigma$ thermally reachable from an input $\rho$ at inverse temperature $\beta$ is bounded by
\begin{equation}\label{Eq:trade-general}     
    \mathcal{NS}(\sigma)\leq \mathcal{NS}(\tau_\beta)+\sqrt{\frac{\beta}{2}\Delta F_\beta(\rho)}.
\end{equation}
In particular, if the Gibbs state is stabiliser, then
\begin{equation}
    \mathcal{NS}(\sigma)\leq \sqrt{\frac{\beta}{2}\Delta F_\beta(\rho)}.
\end{equation}
\end{res}
\begin{sproof}
    The bound follows by comparing $\sigma$ to the Gibbs state $\tau_\beta$, using the continuity of the trace-distance measure of nonstabiliserness, the quantum Pinsker inequality, and the monotonicity of $\Delta F_\beta$ under thermal operations. The full derivation is given in Appendix~\hyperref[App:athermality-nonstabiliserness]{A-3}.
\end{sproof}

Result~\ref{Res:tradeoff} establishes a direct link between nonstabiliserness and quantum thermodynamics: athermality, quantified by the nonequilibrium free-energy difference, upper-bounds the nonstabiliserness that thermal operations can generate. This statement is nontrivial because it applies without assumptions on the microscopic interaction model, control scheme, or system size. In this sense, it identifies athermality as the ``thermodynamic budget'' for nonstabiliserness generation.

However, this bound is only a necessary constraint--it limits what thermal operations can produce, but does not determine when the available athermality can actually be converted into nonstabiliserness. Consequently, Result~\ref{Res:tradeoff} should be interpreted as a fundamental link between two distinct resource theories rather than a tight practical tool. Because the bound is completely agnostic to the microscopic model, it can be quite loose in practice. A simple qubit example illustrates this: consider energy-diagonal inputs with a fixed ground-state population $p$ and two Hamiltonians with the same spectrum, $H_Z=Z$ and $H_T=\tfrac{1}{\sqrt{3}}(X+Y+Z)$. Since the spectrum and populations are the same, the initial nonequilibrium free energy $\Delta F_\beta(\rho)$ is the same in both cases. Nevertheless, as we show below, the $Z$-aligned Hamiltonian does not generate nonstabiliserness from such inputs, whereas the $T$-aligned Hamiltonian can do so once the bath is sufficiently cold. Thus, Result~\ref{Res:tradeoff} captures the thermodynamic price of nonstabiliserness generation, but not the geometric conditions under which that price can actually be paid. 

This motivates our second result: for single qubits, we derive an exact necessary-and-sufficient criterion for when thermal operations generate nonstabiliserness. Specifically, we ask: given a system described by a Hamiltonian \mbox{$H=\hat{\v n}\cdot\v\sigma$}, initially prepared in a stabiliser state and coupled to a heat bath at inverse temperature $\beta$, what are the necessary and sufficient conditions for nonstabiliserness generation?

\begin{res}[Nonstabiliserness generation]\label{Thm:magic-generation}
Magic states are generated from an initial stabiliser state $\rho$ under thermal operations at inverse temperature $\beta$ if and only if 
\begin{equation}\label{Eq:magic-witness-generation}
  \M(\rho):= \max_{q\in I_\beta(p)}\ \max_{\v{s}\in\{\pm1\}^3}
\Bigl(\v{c}_{\max}(q)\, \sqrt{3-h^2_s}+\bigl|2q-1\bigr|\,\bigl|h_s\bigr|\Bigr) > 1.
\end{equation}
where $h_{\v s}:=\v s^{\!\top}\cdot\hat{\v n}$ with $\v{s}\in\{\pm1\}^3$ a sign vector.
\end{res}
\begin{sproof}
    We scan the reachable populations $q$ along the energy axis. For each $q$, coherence is available up to $\v{c}_{\max}(q)$. The stabiliser polytope has flat facets, so each facet is a plane that appears as a straight line in the Hamiltonian-aligned frame. This line encodes a trade-off: increasing the population imbalance $|2q-1|$ forces a decrease in allowable coherence, and increasing coherence forces a decrease in imbalance. If, for any facet and any reachable $q$, the available coherence and height cross that line, the operation has produced nonstabiliserness. See Appendix~\hyperref[App-necessary-sufficient]{B-1} for the formal proof, and Fig.~\ref{F:thermal-cone-plot} for an illustration of the reachable set.
\end{sproof}
Thermal operations induce a natural split between ``easy'' and ``hard'' nonstabiliserness generation. In some instances, the outer maximisation is achieved at the initial population $q=p$, giving a condition that is independent of $\beta$. Temperature then plays no role--only geometry does--and nonstabiliserness can be produced at fixed energy populations by adjusting only the coherence phase, i.e., by applying a time translation $e^{-itH}$, without requiring any population change. In contrast, when the maximiser occurs at  $q\neq p$, temperature becomes decisive: this is hard nonstabiliserness, where tuning $\beta$ enlarges the thermally reachable set and/or increases the available coherence until the stabiliser threshold is crossed.

These observations about Result~\ref{Thm:magic-generation} show that temperature, coherence, and geometry are critical parameters for nonstabiliserness generation. To make this concrete, we consider a simpler scenario in which the temperature dependence and the underlying geometry emerge directly from an explicit criterion. In particular, for states that are diagonal in the energy eigenbasis, Result~\ref{Thm:magic-generation} admits a simple formulation (see Appendix~\hyperref[App-necessary-sufficient]{B-1}):
\begin{equation}\label{Eq:magic-generation-incoheren}
    \M(p,\beta)=\|\hat{\v{n}}\|_{1}\times
\begin{cases}
\bigl|1-2p\,e^{-2\beta}\bigr|, & p<\gamma,\\[4pt]
|2p-1|, & p>\gamma.
\end{cases}
\end{equation}
Nonstabiliserness is generated iff $\M(p,\beta)>1$. Equation~\eqref{Eq:critical-temperature} reveals two distinct regimes: For $p<\gamma$, as the criterion shows, magic states appear only once the inverse temperature crosses a critical value, 
\begin{equation}\label{Eq:critical-temperature}
    \beta_{\text{crt}}> \frac{1}{2}\log\qty(\frac{2\|\hat{\v n}\|_1 p}{\|\hat{\v{n}}\|_1-1}).
\end{equation}
Whereas for $p>\gamma$, the $\beta$–dependence drops out and the condition is purely geometric, nonstabiliserness hinges only on the Hamiltonian direction. In either case, if $H$ is aligned with a pure stabiliser direction, no magic state is generated.

The witness $\M(p,\beta)$ admits operational interpretations that quantify the robustness and thermodynamic requirements of nonstabiliserness generation. When $\M(p,\beta)>1$, the excess $\M-1$ determines the minimum depolarising noise strength needed to destroy the generated nonstabiliserness, providing a direct measure of robustness. For $p<\gamma$, the normalised witness admits a reparameterisation $L$ with the exact identity $L=2(\beta-\beta_{\text{crt}})$, which means it behaves linearly with the inverse-temperature surplus. These interpretations, discussed in Appendix~\hyperref[App:operational-thermodynamic-meanings]{B-2}, transform the witness from a binary indicator into quantitative measures of both operational robustness and thermodynamic resource requirements in the energy-incoherent setting.
\begin{figure}
    \centering
    \includegraphics{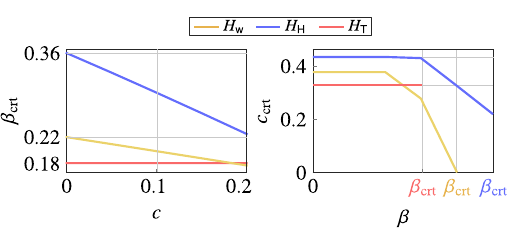}
    \caption{\emph{Critical parameters}. For a fixed population $p=0.3$ and three Hamiltonians with Bloch directions $\hat{\v n}\in\{\tfrac{1}{\sqrt{2}}(1,1,0),\tfrac{1}{\sqrt{6}}(2,1,1),\tfrac{1}{\sqrt{3}}(1,1,1)\}$ (shown as blue $H_\ms{H}$, yellow $H_{\ms w}$, and red $H_\ms{T}$), we illustrate the critical parameters for nonstabiliserness generation. Left panel: critical inverse temperature $\beta_{\rm{crt}}$ versus coherence $c$. The curves decrease monotonically with $c$. Right panel: critical coherence $c_{\rm{crt}}$ versus inverse temperature $\beta$. Lowering the temperature decreases the amount of coherence required to generate nonstabiliserness.
    }
    \label{F:critical-parameters}
\end{figure}

The results discussed here are fundamental: they represent upper bounds on nonstabiliserness generation under energy conservation that do not depend on any particular control scheme. In the Bravyi–Kitaev model~\cite{Bravyi2005}, all the non-stabiliserness is contained in the ancillae, manipulated through Clifford operations. Our results remain operationally relevant as they represent fundamental limits that cannot be surpassed by any protocol respecting the same thermodynamic constraints. In representative cases, most notably energy-diagonal inputs and Hamiltonian direction $H\propto X+Y+Z$ where the closed-form witness in Eq.~\eqref{Eq:magic-generation-incoheren} and the critical inverse temperature in Eq.~\eqref{Eq:critical-temperature} apply, simple Clifford-based routines (Pauli-frame alignment and optional Pauli twirling) together with a single thermal interaction can approach the thermal operations boundary that triggers nonstabiliserness. This illustrates the practical relevance of our bounds, while noting that full code-specific implementation of the thermal interaction is architecture and protocol-dependent.

\section{Optimal regimes}\label{Sec:optimal}
Critical inverse temperatures for general qubit inputs can be obtained analytically once the Hamiltonian is fixed. Coherence never raises, and often lowers, the inverse-temperature threshold required for nonstabiliserness. The relation between critical temperature and coherence depends on how the thermally reachable set meets the stabiliser boundary for a given Hamiltonian direction. When the contact point depends only on populations--specifically for a Hamiltonian aligned with the $\ket{\ms{T}}$ Bloch axis [see Eq.~\eqref{eq:t}]--adding coherence does not change the threshold, and $\beta_{\text{crt}}$ remains flat in $c$ (see the red curve in the left panel of Fig.~\ref{F:critical-parameters}). As the Hamiltonian tilts, the relevant boundary becomes sensitive to coherence: increasing $c$ lowers the critical inverse temperature. More misaligned orientations start with a higher threshold at $c=0$, but gain the most from coherence, yielding the steepest drop (blue and yellow curves in the left panel of Fig.~\ref{F:critical-parameters}).

Next we turn the question around: at a given inverse temperature $\beta$, how much coherence is needed to tip the state outside the stabiliser polytope? For each Hamiltonian direction we define $c_{\text{crt}}(\beta)$ as the smallest coherence that still produces nonstabiliserness at that temperature. The right panel of Fig.~\ref{F:critical-parameters} traces this requirement. For the $H_\ms{T}$ orientation (red), the curve is flat: the contact with the stabiliser boundary is set purely by populations, so changing the temperature does not change how much coherence is needed. Once the inverse temperature reaches the threshold $\beta_{\text{crt}}$, the coherence requirement drops to zero--beyond this point, the Hamiltonian generates nonstabiliserness even without coherence. For tilted directions (blue and yellow), the story is different. As the bath gets colder, the thermally reachable set extends farther along the energy axis, and the required coherence steadily decreases. Each curve eventually meets its own $\beta_{\text{crt}}$; past that vertical mark, $c_{\text{crt}}$ remains at zero, meaning that cooling alone suffices, and no coherence is needed thereafter.

\begin{figure}[t]
    \centering
    \includegraphics{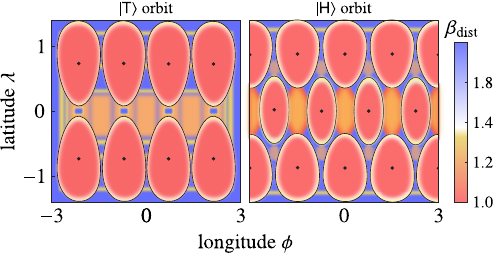}
    \caption{\emph{Distillability landscape over Hamiltonian orientations}. For each Hamiltonian direction $\hat{\v n}$, we consider an initial stabiliser state whose energy-basis parameters are $p = 0.3$ and $c=0.1$. Each point on the equirectangular map (longitude $\phi$, latitude $\lambda$) corresponds to a Hamiltonian direction. Colours show the critical inverse temperature $\beta_{\rm dist}$ at which the maximal thermally reachable fidelity first exceeds the distillation threshold [$\ket{\ms{T}}$ orbit (left) and $\ket{\ms{H}}$ orbit (right)]. The red lobes are the easiest orientations: they are centred on the Clifford‑equivalent directions of the relevant orbit (black dots), eight for $\ms{T}$ and twelve for $\ms{H}$.}
    \label{F:distability-plot}
\end{figure}

\begin{figure*}
    \centering
    \includegraphics{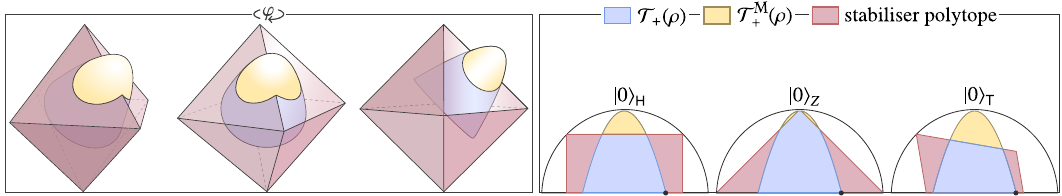}
    \caption{\emph{Reachability under thermodynamic processes}. Achievable states under thermal operations at $\beta = 2$, shown from three viewing angles alongside the stabiliser polytope. Starting from a stabiliser state with $p = 0.5$, $c = 0.25$, and Hamiltonian direction $\hat{\v{n}}_{\ms H} = (\tfrac{1}{\sqrt{2}}, \tfrac{1}{\sqrt{2}}, 0)$, the yellow region marks the produced magic states. Right: for the same initial state (black dot), upper hemispheres of fixed Bloch-ball cross-sections are shown for three Hamiltonians: $\hat{\v{n}}_{\ms H}$, $\hat{\v{n}}_Z = (0, 0, 1)$, and $\hat{\v{n}}_{\ms T} = (\tfrac{1}{\sqrt{3}}, \tfrac{1}{\sqrt{3}}, \tfrac{1}{\sqrt{3}}).$}
    \label{F:thermal-cone-plot}
\end{figure*}

Up to this point we have shown when and how a bath can push stabiliser inputs beyond the stabiliser set, and we identified the temperature–coherence–geometry tradeoffs that control this transition. However, merely being a nonstabiliser state does not guarantee usefulness for computation. What matters operationally is whether the bath can deliver states that cross known universality thresholds. The operational benchmark is magic-state distillation~\cite{Bravyi2005}: the relevant targets are the Clifford orbits of the single-qubit states
\begin{align}
\label{eq:h}
    \ketbra{\ms{H}}&=\frac{1}{2}\,\qty(\iden+\frac{X+Z}{\sqrt{2}}),\\
    \label{eq:t}
    \ketbra{\ms{T}}&=\frac{1}{2}\,\qty(\iden+\frac{X+Y+Z}{\sqrt{3}}).
\end{align}
Bravyi \& Kitaev showed that~\cite{Bravyi2005}, after optimising over single-qubit Clifford unitaries, computational universality is achieved once the overlap with the corresponding orbit exceeds a protocol-dependent threshold: $f_{\rm thr}^{(T)}\approx 0.91$ and $f_{\rm thr}^{(H)}\approx 0.93$, whereas states with $f_{\rm stab}< f < f_{\rm thr}$ exhibit bound nonstabiliserness and remain undistillable~\cite{Campbell2010}. To connect our thermodynamic analysis to this benchmark, we compute, for each Hamiltonian direction, the minimal inverse temperature $\beta_{\rm dist}^{(o)}(\hat{\v n})$ at which the best thermally reachable fidelity with the $o\in\{\ms{T},\ms{H}\}$ orbit crosses $f_{\rm thr}^{(o)}$ (see Appendix~\ref{App:distillability}). The resulting distillability landscape (see Fig.~\ref{F:distability-plot}) complements Result~\ref{Thm:magic-generation} by translating the witness into a universality threshold and reveals a pronounced dependence on Hamiltonian orientation, reinforcing the trends shown in Fig.~\ref{F:critical-parameters}.

The colour scale encodes how cold one must go to cross the distillation threshold: red regions correspond to smaller $\beta_{\rm dist}$ (easier access to distillable nonstabiliserness), while colder colours indicate larger $\beta_{\rm dist}$. The red lobes mark orientations where thermal trajectories intersect the stabiliser boundary in a manner most favourable to the selected target; they are centred on the Clifford-equivalent directions of that orbit (black dots) and inherit their multiplicities (eight for $\ms{T}$, twelve for $\ms{H}$). Moving away from a lobe centre the colour warms, reflecting the increasing inverse temperature required as the Hamiltonian direction becomes less aligned with the target geometry. The shapes and extents of the lobes quantify this tolerance to misalignment and make explicit the strong orientation dependence anticipated by Fig.~\ref{F:critical-parameters}: some directions enable distillable nonstabiliserness at comparatively high temperatures, whereas others require significantly deeper cooling.

A recurring central ingredient in our discussion is the system Hamiltonian. The preceding analysis points to a distinguished choice: $H\propto X+Y+Z$. Geometrically, choosing a fixed Hamiltonian with direction parallel to the $\ket{\ms{T}}$ Bloch axis makes the thermally reachable set intersect the stabiliser boundary in the most favourable way for both our witness and distillation, which is why the blue lobes in Fig.~\ref{F:distability-plot} are centred on the Clifford images of this direction. The $\ms{T}$ orientation minimises the cooling needed for distillable nonstabiliserness and, once the temperature threshold is met, renders coherence nonessential. In Appendix~\ref{App:optimal-Hamiltonian}, we provide a formal proof that this Hamiltonian is thermodynamically optimal for generating nonstabiliserness.

Finally, after asking whether thermal operations can generate magic states from a stabiliser input $\rho$ at inverse temperature $\beta$, we now ask where they can take us—to which nonstabiliser states. This can be determined by constructing the set of states reachable by thermal operations and taking its set difference with the stabiliser polytope, $\mathcal{T}^{\M}_+(\rho):=\mathcal{T}_+(\rho)\setminus\stab$. We call $\mathcal{T}^{\M}_{+}(\rho)$ the nonstabiliser cone. Any state $\rho$ in $\mathcal{T}^{\M}$ satisfies $\M(\rho)>1$. Unlike the witness, a binary indicator of whether nonstabiliserness can be generated, the nonstabiliser cone tells us which magic states are reachable and how large that set is inside the Bloch ball. This suggests a figure of merit: the normalised volume of the nonstabiliser cone, which measures how much of the Bloch ball becomes thermodynamically accessible as nonstabiliserness. Both temperature and the Hamiltonian play a role here. The Hamiltonian sets the direction in which the thermally reachable set grows, while the temperature governs whether it grows or shrinks. For a fixed initial state, the reachable set exhibits two regimes. When $p<\gamma$, its volume grows monotonically with $\beta$: as the bath cools, the thermal state moves toward the ground state, and because the thermal state lies in $\mathcal{T}_+(\rho)$, the reachable set stretches along the energy axis, thus growing the cone and generating nonstabiliserness. Conversely, when $p>\gamma$, the cone volume can shrink as $\beta$ increases until the crossing where $\gamma(\beta)=p$. Beyond this threshold the cone grows, the nonstabiliserness region appears, and its volume increases with further cooling. The Hamiltonian orientation has a similar effect: the closer it is to the optimal direction, the more efficiently it generates magic states, resulting in a larger cone volume (see Fig.~\ref{F:thermal-cone-plot}).

\section{Discussion \& outlook}\label{Sec:discussion}
Does extending the results to higher-dimensional systems reveal anything new? While moving to qudits undeniably enriches stabiliser geometry and Clifford structure, the qubit setting already captures the mechanism relevant here. Under thermal operations~\eqref{Eq:thermal-operations}, the interplay of thermomajorisation, time-translation symmetry, and stabiliser facets is completely exposed in one qubit: the stabiliser boundary is explicit~\eqref{Eq:stabiliser-condition}. Our witness $\M(\rho)$ provides a necessary-and-sufficient test for nonstabiliserness generation, and for energy-diagonal inputs, the criterion collapses to a closed form~\eqref{Eq:magic-generation-incoheren} that gives analytic thresholds~\eqref{Eq:critical-temperature}. Importantly, our results are model-independent: we make no assumptions about interaction strength (weak or strong), structure (local or collective), or duration relative to natural time scales. Therefore, they provide fundamental bounds for any energy-conserving unitary model, regardless of microscopic details.

While we briefly discussed representative cases where Clifford-only logical control can approach these bounds, establishing constructive, code-specific protocols and general tightness is an interesting direction for future work. A natural extension motivated by fault-tolerant quantum computation would be to consider \emph{Clifford-assisted thermal operations}, where both the system control and system-bath interactions are restricted to Clifford operations. Another natural next step is to include a catalyst and explore what advantages catalytic transformations can offer~\cite{Jonathan1999,Datta2022,Patrykreview}. Catalysis enlarges the set of reachable states~\cite{Czartowski2023,Son2024CETO,Son2024H,Czartowski2024}, which should, in turn, enable access to a broader range of magic states, as briefly discussed in Appendix~\ref{App:catalytic}.

Beyond generating nonstabiliserness, an intriguing question is how to detect nonstabiliserness under minimal assumptions~\cite{macedo2025witnessingmagicbellinequalities,junior2025geometricanalysisstabilizerpolytope,Zamora2025,varela2026predictingmagicmeasurements,leone2026unbearablehardnessdecidingmagic}. Similarly, we can ask whether thermodynamic processes can be used for this task. The recent letter~\cite{deOliveiraJunior2025heat} proposes a heat-based certification procedure that is, in principle, applicable to arbitrary quantum resources. Whether this framework can be leveraged to probe nonstabiliserness remains open. 

Our results are directly implementable with current experimental platforms. Nuclear resonance (NMR) systems~\cite{vandersypen2004nmr} and cavity QED setups with superconducting qubits~\cite{Blais2021,Junior2024} have already demonstrated energy-preserving unitaries~\cite{Serra2014,serra2016,cottet2017observing,Micadei2019}, offering precise control over both the initial states and the Hamiltonians. Looking beyond proof-of-principle demonstrations, an important next step is to investigate whether our methods can provide practical advantages for quantum computing within these architectures. 

\section{Acknowledgments}
We thank Ryuji Takagi, Daniel Stilck França \& Jake Xuereb for insightful comments on the manuscript. A. O. J. and J. B. B. acknowledges financial support from the Danish National Research Foundation grant bigQ (DNRF 142), VILLUM FONDEN through a research grant (40864). 
J.Cz. acknowledges support by the start-up grant of the
Nanyang Assistant Professorship at the Nanyang Technological University in Singapore, awarded to Nelly Ng. R. A. M. acknowledges funding from Coordenação de Aperfeiçoamento de Pessoal de Nível Superior – Brasil (CAPES) – Finance Code 001. R. C. acknowledges the Simons Foundation (Grant No. 1023171, R.C.), the Brazilian National Council for Scientific and Technological Development (CNPq, Grants No.403181/2024-0, and 301687/2025-0), the Financiadora de Estudos e Projetos (Grant No. 1699/24 IIF-FINEP) and a guest professorship from the Otto M\o nsted Foundation.

\bibliography{2-citations}

\appendix
\onecolumngrid

\section{General theoretical recap}

In order to make the work self-contained, in this appendix we provide a compact review of the two frameworks used in the paper: Stabiliser theory~(see Ref~\cite{Nielsen2010} for a detailed discussion) \& Thermal operations~(see Ref~\cite{junior2024geometricinformationtheoreticaspectsquantum} for a detailed discussion).

\subsection{Stabiliser theory in a nutshell}\label{App:stabiliser-theory}

The starting point is the $n$-qubit Pauli group $\mathcal P_n = \{\pm 1, \pm i\}\cdot\{\iden,X,Y,Z\}^{\otimes n}$, whose elements are tensor products of single-qubit Pauli operators. A stabiliser state on $n$-qubits is a pure state that is the joint $+1$ eigenstate of an abelian subgroup $S\subset \mathcal{P}_n$ of size $|S| = 2^n$. 

For a single qubit, pure stabiliser states are precisely the eigenstates of the Pauli operators. Writing the Pauli vector $\v \sigma =(X,Y,Z)$ and the unit Bloch vectors $\v x,\v y,\v z$ along the Cartesian axes, the six pure stabiliser states are $\qty{\tfrac{\iden + \v a \cdot \v{\sigma}}{2}}_{\v a \in \{\pm \v x, \pm \v y, \pm \v z\}}$. The stabiliser polytope is then defined as
\begin{equation}
    \stab := \mathrm{conv} \left\{ \frac{\iden + \v a \cdot \v{\sigma}}{2}\right\}_{\v a \in \{\pm \v x, \v y, \v z\}},
\end{equation}
namely the convex mixtures of the six Pauli eigenstates~\cite{Veitch2012,Heinrich2019}. Writing a general state as $\rho=\tfrac{1}{2}\left(\iden + \v r\cdot \v \sigma \right)$, with $\v r =(r_x, r_y, r_z)$ and $\v\sigma =(X,Y,Z)$, membership of $\stab$ is characterised by the single facet inequality
\begin{equation}
    |r_x|+|r_y|+|r_z|\leq 1,
\end{equation}
which compactly captures the eight octahedral facets. The geometry of the stabiliser polytope is shown in Fig.~\ref{F-stabiliser-polytope}. The stabiliser set is closed under the standard free operations of the resource theory: Clifford unitaries $C\in U(2)$ that permute Pauli operators, and measurements in the computational basis; compositions of these are stabiliser operations~\cite{Veitch2012}.

We quantify deviation from stabiliser polytope via the geometric monotone given by the minimal trace distance to the polytope, which we refer to as the nonstabiliserness of $\rho$:
\begin{equation}\label{Eq:app-nonstabiliserness-monotone}
    \mathcal N\mathcal S(\rho) = \min_{\sigma \in \stab} T(\rho, \sigma),
\end{equation}
where $T(\rho, \sigma ) = \tfrac12\| \rho-\sigma\|_1$ is the trace distance~\cite{Nielsen2010}. 

The computational model we consider assumes fault-tolerant access to Clifford dynamics, together with standard stabiliser primitives. Concretely, we allow \emph{(i)} $n$-qubit Clifford unitaries
\begin{equation}
    \mathcal C \ell_n = \{C \in U(2^n)| C P C^\dagger \in \mathcal P_n,  \forall P \in \mathcal P_n\},
\end{equation}
where $\mathcal P_n$ is the $n$-qubit Pauli group; \emph{(ii)} the preparation of $m$ ancillas $\ket{0}^{\otimes m}$; and \emph{(iii)} Pauli measurements. In addition, we assume access to a pool of $n$ noisy single-qubit magic states $\rho(f, \v a)^{\otimes n}$ of the form
\begin{equation}
    \rho(f, \v a) = \frac{\iden + (2f-1)\v a \cdot \v{\sigma}}{2},
\end{equation}
where $\v a \in S^2$ denotes the Bloch vector of a target pure state $\ket{\psi(\v a)}$ and $f = \bra{\psi(\v a)} \rho(f, \v a) \ket{\psi (\v a)}$ is the fidelity; $f=1$ yields $\ket{\psi(\v a)}^{\otimes n}$, while $f=\tfrac12$ gives the maximally mixed state.

Under these assumptions, the computation consists of performing computational-basis measurements on
\begin{equation}
     C (\ketbra{0}^{\otimes m}\otimes \rho(f,\v a)^{\otimes n})C^\dagger.
 \end{equation}
with $C$ a Clifford unitary acting on $n+m$ qubits. The usefulness of $\rho(f,\v a)$ is governed by fidelity thresholds. There exists $f_{\stab}$ such that, if $f \le f_{\stab}$, then $\rho(f,\v a)\in\stab$, in which case the resulting statistics are classically simulable by the Gottesman–Knill theorem~\cite{gottesman1998, Aaronson2004}. Conversely, Bravyi–Kitaev~\cite{Bravyi2005} established that, for target states taken from the Clifford orbits of
\begin{align}
    \ketbra{\ms{T}} &= \frac 1 2 \left(\iden + \frac{X+Y+Z}{\sqrt 3 }\right),\\
    \ketbra{\ms{H}} &= \frac 1 2 \left(\iden + \frac{X+ Z}{\sqrt 2}\right).
\end{align}
there exists a threshold $f_\mathrm{thr}>f_{\stab}$ above which the model becomes computationally universal via magic-state distillation. In particular, universality is attained via the Bravyi-Kitaev protocol~\cite{Bravyi2005} once the orbit-optimised fidelities satisfy
\begin{align}
    \max_{C \in \mathcal C \ell_1} \sqrt{\bra{\ms{T}} C \rho C^\dagger \ket{\ms{T}}} \gtrsim 0.91    ,\\
    \max_{C \in \mathcal C \ell_1} \sqrt{\bra{\ms{H}} C \rho C^\dagger \ket{\ms{H}}} \gtrsim 0.93    .
\end{align}
Moreover, $f_\mathrm{thr} > f_{\stab}$ strictly for all protocols~\cite{Campbell2010}, so states with $f_{\stab}<f<f_\mathrm{thr}$ possess bound nonstabiliserness and remain undistillable. This delineates the operational region relevant to the main text, where thermal processes are shown to move stabiliser inputs beyond $\stab$ and, in favourable regimes, across known distillation thresholds without revisiting the broader stabiliser-theory exposition already provided.

\subsection{Thermal operations in a nutshell}\label{App:coherent-thermal-cones}

In this appendix we provide a review of technical details on the characterisation of states reachable under thermal operations. We follow the convention as set in the main text: a system with Hamiltonian $H$ initially in state $\rho$, a bath with Hamiltonian $H_{\ms E}$ prepared in a Gibbs state $\gamma_{\ms{E}} = \exp(-\beta H_{\ms{E}})$ at inverse temperature $\beta$, and an energy-preserving unitary $U$ satisfying $[U,H+H_{\ms E}]=0$, which induces the map in Eq.~\eqref{Eq:thermal-operations}, namely
\begin{equation}
\label{eq:thermal_ops}
\E(\rho)=\tr_{\ms E}\left[U\left(\rho\otimes\gamma_\ms{E}\right)U^{\dagger}\right].
\end{equation}
The energy-conservation condition encodes the first law of thermodynamics. Moreover, since the heat bath is in thermal equilibrium, every thermal operation $\E$ preserves the system’s Gibbs state, $\E(\gamma)=\gamma$ defined in full analogy to the environment's Gibbs state. Consequently, this Gibbs-preserving property, together with any contractive distance measure $\delta$ yielding $\delta(\rho,\gamma)\geq\delta[\E(\rho),\E(\gamma)]=\delta[\E(\rho),\gamma]$, captures the core content of the second law by formalising the evolution towards thermal equilibrium.

The set of states achievable under Eq.~\eqref{Eq:thermal-operations} is the future thermal cone, denoted $\mathcal{T}_+(\rho)$. Two structural facts allow us to construct the qubit future thermal cone by treating populations and coherences separately and then combining them. First, the diagonal (energy populations) evolves under Gibbs-preserving, energy-conserving stochastic maps and is exactly constrained by thermomajorisation. Second, time-translation symmetry bounds the magnitude of the off-diagonal terms independently of their phase, yielding a tight coherence constraint for any prescribed population change.

For a qubit with ground-state population $p$, coherence $c$, and associated thermal state $\tau_\beta := e^{-\beta H}/\tr(e^{-\beta H})$—the Gibbs state of $H$, which in the energy eigenbasis of $H$ is $\operatorname{diag}(\gamma,1-\gamma)$—the attainable populations and coherences under thermal operations are as follows:

\begin{lem}[Reachable populations~\cite{Lostaglio2018}]\label{Lem:population-evolution}
Under thermal operations at inverse temperature $\beta$, the set of reachable populations is the line segment
\begin{equation} 
\mathcal{P}_\beta(p):=\operatorname{conv}\qty[\qty(p,1-p), \qty(1-\frac{1-\gamma}{\gamma}p,\frac{1-\gamma}{\gamma}p)]. 
\end{equation}
\end{lem}
Projecting $\mathcal{P}_\beta(p)$ onto its first coordinate collects the reachable ground-state populations into the interval
\begin{equation}
I_\beta(p):=\qty[\min\{p,\,q_\star\},\,\max\{p,\,q_\star\}],
\end{equation}
where $q^\star:=1-\tfrac{1-\gamma}{\gamma}p$. We now state the corresponding coherence bound.
\begin{lem}[Reachable coherences~\cite{Lostaglio2015}]\label{Lem:coherence-evolution}
Under thermal operations, the coherence $r$ achievable for a target qubit with population $q$ (from an initial qubit with population $p$ and coherence $c$) is bounded by $r\le\v{c}{\max}(q)$, where
\begin{equation}
\v{c}_{\max}(q)\!:=\!|c| \frac{\sqrt{[q(1-\gamma)\!-\!\gamma(1-p)][p(1-\gamma)\!-\!\gamma(1-q)]}}{|p-\gamma|}.
\end{equation}
\end{lem}
% The interplay between the population and coherence constraints (Lemmas~\ref{Lem:population-evolution} and~\ref{Lem:coherence-evolution}) is illustrated in Fig.~\ref{F-app-future-thermal-cone}, which depicts the qubit future thermal cone $\mathcal{T}_+(\rho)$.
% \begin{figure}
%     \centering
%     \includegraphics{F-app-future-cone.pdf}
%     \caption{\emph{Reachable sets under thermal operations for different Hamiltonian orientations}. 
%     For a fixed bath inverse temperature $\beta = 2$, the three panels show the sets of qubit states reachable under thermal operations, $\mathcal{T}_+(\rho)$, for Hamiltonians whose Bloch vectors point along (a) $(0,0,1)$, (b) $(\tfrac{1}{\sqrt{2}},0,\tfrac{1}{\sqrt{2}})$, and (c) $(\tfrac{1}{\sqrt{3}},\tfrac{1}{\sqrt{3}},\tfrac{1}{\sqrt{3}})$, respectively. We also highlight the magic states $\ket{\ms{H}}$ and $\ket{\ms{T}}$.}
%     \label{F-app-future-thermal-cone}
% \end{figure}
As a last observation, for qubits, the population interval $I_\beta(p)$ [Lemma~\ref{Lem:population-evolution}] together with the coherence bound $r\le \v{c}_{\max}(q)$ [Lemma~\ref{Lem:coherence-evolution}] gives an exact characterisation of $\mathcal{T}_+(\rho)$: for every $q\in I_\beta(p)$ the boundary value $r=\v{c}_{\max}(q)$ is attainable by a thermal operation, and since dephasing in the energy basis is itself a thermal operation and thermal operations are convex (closed under mixing), all intermediate $0\le r\le \v{c}_{\max}(q)$ are attainable at the same $q$. Time-translation covariance implies the coherence phase is freely adjustable, yielding full disks in the Bloch cross-sections.

\subsection{A thermodynamic bound on nonstabiliserness}\label{App:athermality-nonstabiliserness}

In this appendix, we derive an inequality that constrains the nonstabiliserness of any output state that is thermally reachable from a given input. This bound provides a necessary condition on what thermal operations can produce. Importantly, our derivation makes no assumptions about system size or dimension--the result applies equally to single qubits, qudits, and many-body systems.

Consider a quantum system prepared in a state $\rho$ and characterised by a Hamiltonian $H$. A standard athermality monotone in the resource theory of thermodynamics is the nonequilibrium free energy difference, $\Delta F_\beta(\rho):= F_\beta(\rho)-F_\beta(\tau_\beta)$, where $F_\beta(\rho):=\tr(\rho H)-\beta^{-1} S(\rho)$. Importantly, the nonequilibrium free energy difference can be expressed in terms of the quantum relative entropy of $\rho$ with respect to its associated thermal state $\tau_\beta=\exp(-\beta H)/Z$ with $Z=\tr[\exp(-\beta H)]$, i.e., 
\begin{equation}\label{Eq:app-relative-entropy}
    D(\rho\|\tau_\beta) :=\tr[\rho(\log \rho-\log\tau_\beta)] = \beta[F_\beta(\rho)-F_\beta(\tau_\beta)] = \beta \Delta F_\beta(\rho),
\end{equation}
One can verify that $D(\rho\|\tau_\beta)$ [and hence $\Delta F_\beta(\rho)$] is monotonically decreasing under thermal operations. More precisely, if $\mathcal{E}$ be a thermal operation, the the data-processing inequality gives
\begin{equation}\label{Eq:app-monotonicity}
    \Delta F_\beta [\mathcal{E}(\rho)] \leq \Delta F_\beta (\rho).
\end{equation}
This tells us that thermal operations, as free operations, cannot increase the athermality of a state. From now on, we use ``nonequilibrium free energy difference'' and ``athermality'' unchangingly.

As our measure of nonstabiliserness, we use the trace distance to the stabiliser polytope:
\begin{equation}
    \mathcal N\mathcal S(\rho) = \min_{\sigma \in \stab} T(\rho, \sigma) \quad \text{with} \quad T(\rho, \sigma ) = \frac12\| \rho-\sigma\|_1.
\end{equation}
A key property of this measure is captured by the following Lemma:
\begin{lem}[Existence and continuity bound of $\mathcal{NS}$]\label{Lemma:NS}
The minimum in the definition of $\mathcal{NS}(\rho)$ is attained for every state $\rho$. Moreover, for all states $\rho,\omega$,
\begin{equation}\label{Eq:app-lipschitz-rev}
    \bigl|\mathcal{NS}(\rho)-\mathcal{NS}(\omega)\bigr|\ \le\ T(\rho,\omega).
\end{equation}
\end{lem}
\begin{proof}
Since $\stab$ is compact and $T(\cdot,\cdot)$ is continuous, a minimiser exists. For the bound~\eqref{Eq:app-lipschitz-rev}, let $\sigma_\omega\in\stab$ be a minimiser for $\mathcal{NS}(\omega)$, so that $\mathcal{NS}(\omega)=T(\omega,\sigma_\omega)$. By the triangle inequality,
\begin{equation}
\mathcal{NS}(\rho)=\min_{\sigma\in\stab}T(\rho,\sigma) \le T(\rho,\sigma_\omega) \le T(\rho,\omega)+T(\omega,\sigma_\omega) = T(\rho,\omega)+\mathcal{NS}(\omega).
\end{equation}
Swapping $\rho$ and $\omega$ gives 
\begin{equation}\label{Eq:app-athermality-almost}
    \mathcal{NS}(\omega)\le T(\rho,\omega)+\mathcal{NS}(\rho),
\end{equation}
and combining the two inequalities gives~\eqref{Eq:app-lipschitz-rev}.
\end{proof}

Applying Lemma Eq~\ref{Lemma:NS} with $\omega = \tau_\beta$ gives a simple relation:
\begin{equation}\label{Eq:app-NS-bound-with-gamma}
    \mathcal{NS}(\rho) \leq \mathcal{NS}(\tau_\beta) + T(\rho, \tau_\beta)
\end{equation}
We now connect the trace distance to athermality via the quantum Pinsker inequality:
\begin{equation}\label{Eq:app-quantum-pinsker}
    T(\rho,\tau_\beta) \leq \sqrt{\frac{1}{2}D(\rho\|\tau_\beta)}
\end{equation}
Substituting this into \eqref{Eq:app-NS-bound-with-gamma} gives our central bound for a single state:
\begin{equation}\label{Eq:app-bound-single-state}
\mathcal{NS}(\rho) \leq \mathcal{NS}(\tau_\beta) + \sqrt{\frac{\beta}{2}\Delta F_\beta(\rho)}.
\end{equation}
This inequality shows that the deviation of $\mathcal{NS}(\rho)$ from its equilibrium value $\mathcal{NS}(\tau_\beta)$ is constrained by the state's athermality. Thus, the initial athermality upper-bounds the nonstabiliserness of every output state thermally reachable from $\rho$. Let $\sigma = \mathcal{E}(\rho)$ be any state reachable from $\rho$ by a thermal operation $\mathcal{E}$. Applying the bound~\eqref{Eq:app-trade} to the output $\sigma$ and using the monotonicity of athermality \eqref{Eq:app-monotonicity}, we obtain the desired input-output constraint:
\begin{equation}\label{Eq:app-trade}
    \mathcal{NS}(\sigma) \leq \mathcal{NS}(\tau_\beta) + \sqrt{\frac{\beta}{2}\Delta F_\beta(\rho)}
\end{equation}
Thus, the initial athermality $\Delta F_\beta(\rho)$ upper-bounds the nonstabiliserness of every possible output state thermally reachable from $\rho$.

A particularly instructive special case occurs when the Gibbs state $\tau_\beta$ itself is a stabiliser state, implying $\mathcal{NS}(\tau_\beta) = 0$. In this case, the constraint simplifies to
\begin{equation}
    \mathcal{NS}(\sigma) \leq \sqrt{\frac{\beta}{2}\Delta F_\beta(\rho)}.
\end{equation}
This gives a direct quantitative limitation: to generate an output with nonstabiliserness $\mathcal{NS}(\sigma) \geq m$ via thermal operations, the input $\rho$ must carry at least $\Delta F_\beta(\rho) \geq \frac{2}{\beta}m^2$ of athermality. Therefore, a significant amount of nonstabiliserness can only be produced from a correspondingly nonequilibrium input.

It is important to emphasize that the primary value of the bound in Eq.~\eqref{Eq:app-trade} lies in establishing a fundamental, dimension-independent connection between two distinct resource theories. Because our derivation is completely agnostic to the microscopic details of the Hamiltonian, the dimension of the Hilbert space, and the specific thermal operation applied, it acts as a universal thermodynamic witness rather than a tight state-conversion theorem. In many physical regimes, particularly at low temperatures or for highly misaligned Hamiltonians, this bound will inevitably be loose. Nevertheless, it is a natural and important question to ask how tight this necessary condition can be, and specifically, how much gap is introduced by the mathematical relaxation in our derivation (namely, the quantum Pinsker inequality). To assess the tightness of our bound, we now consider an explicit single-qubit example where the geometric conditions are perfectly optimal, allowing us to isolate and quantify the exact analytical slack.

Consider a qubit system governed by a Hamiltonian aligned with the optimal magic generation direction, $H = -\frac{\epsilon}{\sqrt{3}}(X+Y+Z)$. The corresponding thermal state is
$\tau_\beta = \tfrac{1}{2}\left(\iden + t_0 \frac{X+Y+Z}{\sqrt{3}}\right)$ with $\quad t_0 = \tanh(\beta\epsilon)$. We choose a specific inverse temperature such that $t_0 = \tfrac{1}{\sqrt{3}}$. At this temperature, the thermal state $\tau_\beta$ lies exactly on a facet of the stabiliser octahedron, implying that its initial nonstabiliserness is zero, $\mathcal{NS}(\tau_\beta) = 0$.

Now, consider a family of states parametrised by $t \in [t_0, 1]$ along the same axis: $\rho_t = \tfrac{1}{2}\left(\iden + t \frac{X+Y+Z}{\sqrt{3}}\right)$ Since all states in this family commute, their analysis is particularly straightforward. Let us evaluate the two inequalities that constitute our bound. First, we consider the geometric step from Lemma~\ref{Lemma:NS}. For $t \geq t_0 = \tfrac{1}{\sqrt{3}}$, the nonstabiliserness of $\rho_t$ is simply given by its trace distance to the boundary of the stabiliser polytope along this ray, which is exactly $\tau_\beta$. Thus, we have: $\mathcal{NS}(\rho_t) = T(\rho_t, \tau_\beta) = \tfrac{t - t_0}{2}$. This shows that the geometric inequality $\mathcal{NS}(\rho_t) \leq \mathcal{NS}(\tau_\beta) + T(\rho_t, \tau_\beta)$ is perfectly saturated (i.e., it holds as an exact equality).

Next, we evaluate the entropic step, which relies on the quantum Pinsker inequality~\eqref{Eq:app-quantum-pinsker}. Because $\rho_t$ and $\tau_\beta$ commute, the quantum relative entropy reduces to the classical Kullback-Leibler divergence between their eigenvalues: $D(\rho_t \| \tau_\beta) = \tfrac{1+t}{2} \log \tfrac{1+t}{1+t_0} + \tfrac{1-t}{2} \log \tfrac{1-t}{1-t_0}$. To understand the behaviour near equilibrium, we consider a small deviation from the thermal state, $t = t_0 + \delta$ with $\delta \ll 1$. Taylor expanding the relative entropy gives:
\begin{equation}
    D(\rho_t \| \tau_\beta) = \frac{\delta^2}{2(1-t_0^2)} + \mathcal{O}(\delta^3).
\end{equation}
Substituting this expansion into the right-hand side of Pinsker's inequality, we find:
\begin{equation}
    \sqrt{\frac{1}{2}D(\rho_t\|\tau_\beta)} \approx \sqrt{\frac{\delta^2}{4(1-t_0^2)}} = \frac{\delta}{2\sqrt{1-t_0^2}}
\end{equation}
Comparing this to the exact trace distance $T(\rho_t, \tau_\beta) = \tfrac{\delta}{2}$, the gap introduced by Pinsker's inequality is asymptotically given by the multiplicative factor:
\begin{equation}
    \frac{\sqrt{\frac{1}{2}D(\rho_t\|\tau_\beta)}}{T(\rho_t, \tau_\beta)} \approx \frac{1}{\sqrt{1-t_0^2}}.
\end{equation}
For our specific choice of $t_0 = \tfrac{1}{\sqrt{3}}$, this factor evaluates to $\sqrt{3/2} \approx 1.225$. Therefore, for states close to the thermal state along this optimal ray, the full bound evaluates to $\mathcal{NS}(\rho_t) \lesssim 1.225 \times T(\rho_t, \tau_\beta)$. This demonstrates that the bound is physically meaningful and close to tight, overestimating the true nonstabiliserness by a mere $22.5\%$ multiplicative gap in this regime.

\section{Proofs of the main result}
\subsection{Necessary \& Sufficient condition} \label{App-necessary-sufficient}

We begin with a geometric reformulation that characterises when the future thermal cone lies entirely within the stabiliser polytope. By rotating into the Hamiltonian-aligned frame, we may restrict our considerations to circular cross-sections of the Bloch sphere without loss of generality, thus reducing the problem to a support-function optimisation, separating population and coherence contributions. The resulting inequality is stated below and used in our bounds. It is established via two lemmas and a theorem:

\begin{lem}[Rotated-frame $\ell_1$ support over circles]\label{Lem:circle-support}
Let $H=\hat{\v{n}}\cdot\v\sigma$ be a qubit Hamiltonian and let a rotation $R\in SO(3)$ send $\hat{\v{n}}\mapsto\hat{\v z}$. Set $M:=R^{-1}=R^\top$ with columns $\v m_1,\v m_2,\v m_3$. For fixed $z'\in\mathbb R$ and $r_\perp\ge 0$, define $\v r'(\phi)=(r_\perp\cos\phi,r_\perp\sin\phi,z')$. Then
\begin{equation}
    \max_{\phi\in[0,2\pi)}\bigl\|M\,\v r'(\phi)\bigr\|_1
=\max_{\v{s}\in\{\pm1\}^3}\left(r_\perp\,\alpha_s+|z'|\,|h_s|\right),
\end{equation}
where $\alpha_s:=\sqrt{(s^\top\v m_1)^2+(s^\top\v m_2)^2}$ and $h_s:=s^\top\v m_3$. Since $M$ is orthogonal and $\v{s}\in\{\pm1\}^3$, we have $\alpha^2_s + h^2_s = 3$.
\end{lem}
\begin{proof}
    Using the support-function identity $\|x\|_1 = \max_{s\in\{\pm 1\}} s^{\top} x$, we have $\|M \v r'(\phi)\|_1 = \max_{\v{s}\in\{\pm1\}^3} s^\top M \v r'(\phi)$. Thus, it follows that 
\begin{equation}
    \max_{\phi\in[0,2\pi)}\|M \v r'(\phi)\|_1 =\max_{\phi}\max_{\v{s}\in\{\pm1\}^3} s^\top M \v r'(\phi) =\max_{\v{s}\in\{\pm1\}^3}\max_{\phi} s^\top M \v r'(\phi),
\end{equation}
where we interchange the order of maximisation since $s$ ranges over a finite set and, for fixed $s$, $\phi\mapsto s^\top M \v r'(\phi)$ is continuous on the compact interval $[0,2\pi]$ (so both maxima are attained). Writing $M$ in terms of its columns, $M=\v m_1\ \v m_2\ \v m_3$ and $h_s:=s^\top \v m_3$, we obtain for fixed $s$
\begin{equation}
    s^\top M \v r'(\phi)=r_\perp\bigl(s^\top\v m_1\cos\phi+s^\top\v m_2\sin\phi\bigr)+z'h_s.
\end{equation}
Let $\alpha_s:=\sqrt{(s^\top\v m_1)^2+(s^\top\v m_2)^2}$. If $\alpha_s>0$, set $\tan\varphi_s=\frac{s^\top\v m_2}{s^\top\v m_1}$ so that $ s^\top M \v r'(\phi)=r_\perp\alpha_s\cos(\phi-\varphi_s)+z'h_s$, and hence $\max_{\phi}s^\top M \v r'(\phi)=r_\perp\alpha_s+z'h_s$ (attained at $\phi=\varphi_s$). If $s^\top\v m_1=s^\top\v m_2=0$, then $\alpha_s=0$ and the $\phi$-dependence vanishes, giving $\max_{\phi}s^\top M \v r'(\phi)=z'h_s$. In either case, $\max_{\phi}s^\top M \v r'(\phi)=r_\perp\alpha_s+z'h_s$. Since $\alpha_{-s}=\alpha_s$ and $h_{-s}=-h_s$, maximising over $s$ yields
\begin{equation}
    \max_{\phi}\|M \v r'(\phi)\|_1
=\max_{\v{s}\in\{\pm1\}^3}\bigl(r_\perp\alpha_s+|z'|\,|h_s|\bigr),
\end{equation}
\end{proof}

\begin{lem}[Single-$q$ stabiliser-containment test]\label{Lem:single-q}
Let $q\in[0,1]$ and let $r_\perp(q)$ be an admissible transverse radius at population $q$. In the rotated frame above, the entire circle
\begin{equation}
    \bigl\{\,(x',y',z')=(r_\perp(q)\cos\phi,\,r_\perp(q)\sin\phi,\,2q-1):\ \phi\in[0,2\pi)\,\bigr\}
\end{equation}
lies in the stabiliser polytope $\stab_R:=\{\v r':\|M\v r'\|_1\le 1\}$ iff
\begin{equation}
    \max_{\v{s}\in\{\pm1\}^3}\Bigl(\v{c}_{\max}(q)\,\alpha_s+\bigl|2q-1\bigr|\,\bigl|h_s\bigr|\Bigr)\ \le\ 1
\end{equation}
with $\v{c}_{\max}(q)$ defined in Lemma \ref{Lem:coherence-evolution}.
\end{lem}
\begin{proof}
Apply Lemma~\ref{Lem:circle-support} with $r_\perp=r_\perp(q)$ and $z'=2q-1$, and use the definition of $\stab_R$ as an $\ell_1$-ball under $M$.
\end{proof}
\begin{thm}[TO no-stabiliserness criterion]\label{thm:TO-no-magic}
Let $I_\beta(p)$ be the TO-reachable population interval at inverse temperature $\beta$ starting from population $p$. Then the future thermal cone $\mathcal T_+(\rho)$ is contained in the stabiliser polytope iff
\begin{equation}
     \max_{q\in I_\beta(p)}\ \max_{\v{s}\in\{\pm1\}^3}
\Bigl(\v{c}_{\max}(q)\,\alpha_s+\bigl|2q-1\bigr|\,\bigl|h_s\bigr|\Bigr) \le 1.
\end{equation}
\end{thm}

\noindent
This proves Result~\ref{Thm:magic-generation}. Two immediate corollaries show how the general condition simplifies in special cases. The first concerns energy-incoherent input states:
\begin{cor}[Energy-incoherent states]
For an energy-incoherent input state ($r_\perp = 0$), Theorem \ref{thm:TO-no-magic} simplifies to:
\begin{equation}
    |z'| \|\hat{\v{n}}\|_1 \leq 1.
\end{equation}
Since $\v{m}_3 = R^{-1} \hat{\v{z}} = \hat{\v{n}}$, we have $\max_s |h_s| = \|\v{m}_3\|_1 = \|\hat{\v{n}}\|_1$.
\end{cor}
The second Corollary applies when the Hamiltonian is aligned with a Pauli axis:
\begin{cor}[Pauli-aligned special case]
When $\hat{\v{n}}$ is aligned with a Pauli axis, $R$ is a signed permutation matrix. In this case, the worst azimuth occurs at $|x'| = |y'|$, and the condition reduces to checking the $\pi/4$ section:
\begin{equation}
    \sqrt{2}|u| + |z'| \leq 1,
\end{equation}
where $x'=y'=u$.
\end{cor}

\subsection{Operational and thermodynamic meanings of the magic witness}\label{App:operational-thermodynamic-meanings}

The magic-generation witness $\M(p,\beta)$ provides a sharp boundary between stabiliser and non-stabiliser states achievable via thermal processes, with magic generation occurring precisely when $\M(p,\beta)>1$. While this binary condition determines whether magic states can be generated, examining how the witness behaves relative to this threshold reveals deeper physical insights about the nature of magic generation. By transforming the witness to measure its deviation from the magic-generation threshold, we obtain quantities that characterise both the geometric and thermodynamic aspects of the process.

For energy-diagonal inputs, the witness takes the form
\begin{equation}
\M(p,\beta) = \|\hat{\v{n}}\|_{1} \times
\begin{cases}
\bigl|1 - 2p\, e^{-2\beta}\bigr|, & p < \gamma, \\[4pt]
|2p - 1|, & p > \gamma,
\end{cases}
\end{equation}
where $\|\hat{\v{n}}\|_1\in[1,\sqrt{3}]$ is the $\ell_1$-norm of the Hamiltonian direction in Bloch coordinates. The witness is bounded by this orientational ceiling, $\M(p,\beta)\le \|\hat{\v{n}}\|_1$, with equality approached as $\beta\to\infty$ in the $p<\gamma$ branch.
\smallskip

There are two natural transformations which endow the stabiliserness witness $\M$ with natural interpretations. The first and simplest transformation is the additive shift $\M(p,\beta)-1$. Geometrically, since the single-qubit stabiliser polytope corresponds to the $\ell_1$-ball in Bloch space, this quantity measures the facet-violation margin: when positive, it quantifies how far the thermally reachable set exceeds the boundary of the stabiliser octahedron. This geometric interpretation has direct operational significance. Under isotropic depolarising noise $\mathcal{D}_t(\rho)=(1-t)\rho+t\,\tfrac{\iden}{2}$, which shrinks Bloch vectors by $(1-t)$, the minimal noise strength that eliminates all stabiliserness from thermally reachable states is
\begin{equation}
t^*(p,\beta)=\max\Bigl\{\,0,1-\frac{1}{\M(p,\beta)}\Bigr\}
=
\begin{cases}
0, & \M(p,\beta)\le 1,\\[4pt]
\dfrac{\M(p,\beta)-1}{\M(p,\beta)}, & \M(p,\beta)>1.
\end{cases}
\end{equation}
Thus, the magnitude of $\M-1$ directly quantifies the robustness of magic generation to white noise, providing a practical measure of how much environmental noise the process can tolerate while still producing useful states.

In the $p<\gamma$ regime, the temperature dependence reveals additional structure. Letting $\beta_{\mathrm{crt}}$ satisfy $\M(p,\beta_{\mathrm{crt}})=1$ and defining $\delta\beta:=\beta-\beta_{\mathrm{crt}}\ge 0$, and noting that above threshold we have $|1-2p e^{-2\beta}|=1-2p e^{-2\beta}$, we obtain
\begin{equation}
\M(p,\beta)-1
=
\bigl(\|\hat{\v{n}}\|_1 - 1\bigr)\,\bigl(1 - e^{-2\delta\beta}\bigr).
\end{equation}
This shows that the facet-violation margin increases monotonically with the inverse-temperature surplus $\delta\beta$, saturating at $\|\hat{\v{n}}\|_1-1$ as $\beta\to\infty$. It is important to note that $\M-1$ represents a witness margin rather than a metric distance---it should not be identified with geometric distances such as the trace distance to the stabiliser polytope.
\smallskip

Logarithmic transformation provides a complementary perspective. The quantity $\log\M$ captures multiplicative overshoot and near threshold behaves similarly to $\M-1$, but lacks a desirable property---linear temperature dependence. A more insightful approach normalises by the orientational ceiling:
\begin{equation}
L(p,\beta):=-\log\left(1 - \frac{\M(p,\beta)}{\|\hat{\v{n}}\|_1}\right) +\log\left(1 - \frac{1}{\|\hat{\v{n}}\|_1}\right), \qquad \|\hat{\v{n}}\|_1>1.
\end{equation}
For $p<\gamma$ and any $\beta\ge \beta_{\mathrm{crt}}$, this yields the exact identity $L(p,\beta)=2(\beta-\beta_{\mathrm{crt}})$, transforming the threshold condition into a linear thermometer that directly measures how far the bath is cooled beyond criticality. The requirement $\|\hat{\v{n}}\|_1>1$ excludes the trivial axis-aligned case where no temperature leverage is available.

When $p>\gamma$, the witness becomes temperature-independent with $\M(p,\beta)=\|\hat{\v{n}}\|_1|2p-1|$. In this purely geometric regime, only the additive shift $\M-1$ remains meaningful, capturing facet overshoot and determining robustness via $t^*=1-1/\M$ when $\M>1$, while logarithmic transformations offer no additional thermal insight.

\section{Distillability (inverse temperature) bounds \label{App:distillability}}

We now turn to a complementary question. Given an initial stabiliser state and a bath at an inverse temperature $\beta$ when can the resulting reachable states be useful for magic-state distillation? In other words, how hot can the bath be while still allowing the production of states whose fidelities with the relevant Clifford orbits exceed the known distillation thresholds. 

Since thermal operations are axially symmetric around the Hamiltonian axis $\hat{\v n}$, the phase of the coherence can always be aligned with any given pure target direction $\v u$ on the Bloch sphere. Let $u_z = |\hat{\v n}. \v u|$ and $u_\perp = \sqrt{1-u^2_z}$ be the decomposition of $\v u$ into energy-axis and transverse components in the rotated frame. For a state with a population $q$ (so $z' = (2q-1)$) and coherence amplitude $r \leq \v{c}_{\max}$, the maximal overlap with $\v u$ occurs when their phases coincide, giving
\begin{equation}
    \max_{\text{phase}} \v r. \v u = r u_\perp +|2q-1|u_z.
\end{equation}
The corresponding state fidelity with the pure state $\ket  u$ is $F(\rho, \ket{u}) = \tfrac12(1+\v r \cdot \v u)$, so the maximal achievable fidelity with that direction reads
\begin{equation}\label{Eq:app-fidelity-beta}
    F^{\star}(\beta,\v u) = \frac{1}{2}\qty[1+\max_{q \in I_\beta(p)} (\v{c}_{\max}(q) u_\perp + |2q-1| u_z)].
\end{equation}
This expression is the fidelity analogue of the support-function bound appearing in Theorem~\ref{thm:TO-no-magic}.

Let $\mathcal{O}_T$ and $\mathcal{O}_H$ denote the Clifford orbits of the magic states $\ket{\ms{T}}$ and $\ket{\ms{H}}$, respectively. Each orbit consists of finitely many Bloch directions $\v u$ related by signed permutations of $\tfrac{1}{\sqrt{3}}(1,1,1)$ (for $\ms{T}$) or $\tfrac{1}{\sqrt{2}}(1,0,1)$ (for $\ms{H}$). The best fidelity attainable at temperature $\beta$ is obtained by maximising Eq.~\eqref{Eq:app-fidelity-beta} over the corresponding orbit,
\begin{equation}
    F^{\star}_\mathcal{O}(\beta):=\max_{u\in\mathcal{O}_o} F^{\star}(\beta, \v u) \quad \text{with} \quad o\in\{\ms{T},\ms{H}\}.
\end{equation}
Universal quantum computation is achievable whenever this fidelity exceeds the threshold of a chosen distillation protocol, i.e., $F^{\star}_o(\beta) \geq f_{\text{thr}}^{(o)}$. In our numerical illustrations we use the benchmark values reported for the Bravyi--Kitaev protocols~\cite{Bravyi2005}, namely $f_{\text{thr}}^{(\ms{T})}\approx 0.91$ and $f_{\text{thr}}^{(\ms{H})}\approx 0.93$, but the definition of $\beta_{\text{dist}}^{(o)}$ applies unchanged for any protocol-dependent choice of $f_{\text{thr}}^{(o)}$. The smallest inverse temperature satisfying the inequality defines the distillability bound $\beta_{\text{dist}}^{(o)}$ for the orbit $o$:
\begin{equation}
    \beta_{\text{dist}}^{(o)} :=\text{inf}\qty{\beta \geq 0 : F^{\star}_o(\beta) \geq f_{\text{thr}}^{(o)}
}.
\end{equation}
Since $I_\beta(p)$ and $\v{c}_{\max}(q)$ are monotonically increasing with $\beta$, the map $\beta \mapsto F^{\star}_o(\beta)$ is monotonically increasing, so $\beta_{\text{dist}}^{(o)}$ can be obtained by a simple one-dimensional root-finding procedure.

For energy-incoherent states, the coherence term vanishes and Eq.~\eqref{Eq:app-fidelity-beta} simplifies to
\begin{equation}
    F^{\star}(\beta,\v u) = \frac{1}{2}\qty(1+\max_{q \in I_\beta(p)} |2q-1| \max_{\v u \in \mathcal{O}_o} |\hat{\v{n}}.\v u|),
\end{equation}
where the geometric factors are
\begin{equation}
    \max_{\v u \in \mathcal{O}_T} |\hat{\v{n}}.\v u| = \frac{\|\hat{\v n}\|_1}{\sqrt{3}} \quad , \quad  \max_{\v u \in \mathcal{O}_H} |\hat{\v{n}}.\v u| = \frac{|n|_{(1)}+|n|_{(2)}}{\sqrt{2}},
\end{equation}
with $|n|_{(1)} \geq |n|_{(2)} \geq |n|_{(3)}$ the ordered components of $\v n$. When $\hat{\v n}$ is Pauli-aligned, the maximising directions coincide with the coordinate axes, and the extremum over $q$ usually occurs at interval endpoints $q=p$ or $q=q^{\star}$. Finally, for the optimal Hamiltonian $\hat{\v n} \propto (1,1,1)$, the relevant $\ms{T}$-orbit direction has $u_z = \tfrac{1}{\sqrt{3}}$ and $u_\perp = \sqrt{\tfrac{2}{3}}$, giving the expression:
\begin{equation}
    F_{\ms{T}}^{\star}(\beta) = \frac{1}{2}\qty[1+\max_{q\in I_\beta(p)}\qty(\sqrt{\frac{2}{3}}\v{c}_{\max} (q) + \frac{1}{\sqrt{3}}|2q-1|)]
\end{equation}
from which $\beta_{\text{dist}}^{(\ms{T})}$ follows immediately by solving $F^{\star}_{\ms{T}}\qty(\beta_{\text{dist}}^{(T)}) = 0.91$.

\section{Optimal Hamiltonian}\label{App:optimal-Hamiltonian}

To find the Hamiltonian that maximises the generation of nonstabiliserness under thermal operations, it suffices to analyse the distinguishability of the output state $\mathcal{E}(\rho)$ of a thermal operation from the set of stabiliser states. This naturally leads to the optimisation problem
\begin{equation}\label{Eq:optimisation-problem-optimal-H}
    \max_{H, \|H\|=1} \min_{\sigma \in \stab} \|\mathcal{E}(\rho)-\sigma\|_1
\end{equation}

Without loss of generality, we consider the Hamiltonian $H = \hat{\v n} \cdot \v \sigma$, with $\norm{H} = 1$. Let $\mathcal{E}$ denote a thermal operation and let $\rho$ be the initial state of the system, assumed to be energy-incoherent in the eigenbasis of $H$ described by populations $(p, 1-p)$. The final state after a thermal operation is $\mathcal{E}(\rho) = \tfrac12(\iden + \v r\cdot \v \sigma)$, where $\v r = m \hat{\v n}$, with no transverse components (i.e., no coherence). For a qubit, the thermomajorisation bound on the extreme reachable ground-state population, starting from an initial diagonal state $\rho$ with ground-state population $p$, is $q^{\star} = 1-pe^{-2\beta}$. Hence, the largest achievable Bloch vector length along the direction $\hat{\v n}$ is $m=|2q^\star-1|=|1-2pe^{-2\beta}|$. Crucially, $m$ depends only on $(p,\beta)$ and the fixed spectral norm of $H$ rather than on the orientation $\hat{\v n}$. 

The stabiliser polytope is the octahedron $\{\v s\in \mathbbm{R}^3:\|\v s\|_1\equiv |s_x|+|s_y|+|s_z| \leq 1 \}$. For qubits, the trace distance between states coincides with the Euclidean distance between their Bloch vectors $\|\rho-\sigma\|_1 = \|\v r-\v s\|_2$. Consequently, the optimisation problem in Eq.~\eqref{Eq:optimisation-problem-optimal-H} reduces to evaluating:
\begin{equation}
    V(m) = \max_{\hat{\v{n}}:\|\hat{\v{n}}\|_2 = 1} \min_{\|\v s\|_1\leq1}\|m\hat{\v{n}}-\v s\|_2,
\end{equation}
where $m\geq 0$ and the vectors lie in $\mathbbm{R}^3$. Since the trace norm for qubits equals the Euclidean distance between their Bloch vectors, our optimisation becomes purely geometric: we seek the unit direction $\hat{\v{n}}$ that maximises the Euclidean distance from the point $m \hat{\v{n}}$ to the $\ell_1$-ball: $B_1:=\{\v s\in\mathbbm{R}^3:\|\v s\|_1\leq1\}$. Thus,
\begin{equation}
    f(\hat{\v{n}}):=\min_{\|\v s\|_1 \leq 1}\|m \hat{\v{n}} - \v s\|_2 = \operatorname{dist}(m \hat{\v{n}}, B_1),
\end{equation}
where $\operatorname{dist}(\v x,C)$ denotes the Euclidean distance from $x$ to $C$. Now, we maximise $f(\hat{\v{n}})$ over all unit vectors $\hat{\v{n}}$. To this end, we reformulate the distance via convex duality. For any closed convex set $C$ containing the origin the following holds
\begin{equation}
    \operatorname{dist}(\v x, C) = \max_{\|\v y\|_2\leq 1} \qty(\v y^{\ms T} \v x - \underset{\v z \in C}{\operatorname{sup}\v y^{\ms T}} \v z).
\end{equation}
For $C = B_1$, we obtain $\operatorname{sup}_{\v z \in B_1} \v y^{\ms T} \v z = \|\v y\|_\infty$. Consequently, $f(\hat{\v n}) = \max_{\|\v y\|_2\leq 1} (m \v y^{\ms T}\hat{\v n} - \|\v y \|_\infty)$. We now turn to the computation of the function $V(m)$:
\begin{align}
    V(m) = \max_{\hat{\v{n}}:\|\hat{\v{n}}\|_2 = 1} f(\hat{\v n}) = \max_{\hat{\v{n}}:\|\hat{\v{n}}\|_2 = 1} \max_{\|\v y\|_2\leq 1}(m \v y^{\ms T}\hat{\v n} - \|\v y \|_\infty) =\max_{\|\v y\|_2\leq 1}\max_{\hat{\v{n}}:\|\hat{\v{n}}\|_2 = 1} (m \v y^{\ms T}\hat{\v n} - \|\v y \|_\infty).
\end{align}
We may interchange the order of maximisation over $\hat{\v n}$ and $\v y$ because both feasible sets are compact and the objective is continuous and bilinear in $(\hat{\v n}, \v y)$. This lets us first solve $\max_{\|\hat{\v n}\|_2=1} \v y^{\ms T}\hat{\v n}=\|\v y\|_2$ for fixed $\v y$. Let $t=\|\v y\|_2 \in [0,1]$, for any $\v y$ with $\|\v y\|_2 = t$, it follows that $\|\v y\|_\infty \geq \tfrac{t}{\sqrt{3}}$. Indeed, if all $|y_i| < t/\sqrt{3}$, then $\|\v y\|_2 < t$, a contradiction. Hence, $m\|\v y\|_2 - \|\v y\|_\infty \leq t(m - \tfrac{1}{\sqrt{3}})$. Therefore, $V(m) \leq \max_{0\leq t \leq 1} t\qty(m-\frac{1}{\sqrt{3}})$. This leads to two cases. If $m\leq \tfrac{1}{\sqrt{3}}$, the maximum is $0$, attained at $t=0$. Taking $\v y = \v 0$ shows that $V(m)\geq 0$, and hence, $V(m) = 0$. In the second case, when $m\geq \tfrac{1}{\sqrt{3}}$, the maximum is attained at $t=1$, giving $V(m)\leq m -\frac{1}{\sqrt{3}}$. Choosing $\v y = \frac{1}{\sqrt{3}}(1,1,1)$ yields $\|\v y\|_2 = 1$ and $\| \v y\|_\infty = \tfrac{1}{\sqrt{3}}$, which shows that $V(m) = m -\frac{1}{\sqrt{3}}$. Combining the two cases, we obtain:
\begin{align}
    V(m) = \begin{cases} 0 \quad \quad &\text{if} \quad m  \leq \frac{1}{\sqrt{3}} \\
    m-\frac{1}{\sqrt{3}} \quad &\text{if} \quad m > \frac{1}{\sqrt{3}}
    \end{cases}
\end{align}
From the equality above (attained at $\v y = \tfrac{1}{\sqrt{3}}(1,1,1)$), the inner maximisation $\max_{\|\hat{\v n}\|_2 = 1} \v y^{\ms T} \hat{\v n}$ is achieved by choosing $\hat{\v n}$ parallel to $\v y$. Hence, an optimal axis is
$\hat{\v n}_\star= \tfrac{1}{\sqrt{3}}(1,1,1)$ and, by the symmetry of the $\ell_1$-ball, any sign pattern $\hat{\v n}_\star= \tfrac{1}{\sqrt{3}}(\pm 1,\pm 1,\pm1)$ is also optimal. Consequently, the optimal Hamiltonian is
\begin{equation}
    H_\star = \hat{\v n}_\star\cdot \v\sigma = \frac{X+Y+Z}{\sqrt{3}}.
\end{equation}

\section{Remarks on catalytic extension \label{App:catalytic}}

We briefly outline how our framework extends to catalytic scenarios, providing a high-level overview rather than detailed analysis. The main text establishes a necessary-and-sufficient condition for magic state generation under thermal operations (TO), together with the future thermal cone of magic $\mathcal{T}_+^{\mathcal{M}}(\rho)$ and its normalised volume $\mathcal{V}_+^{\mathcal{M}}(\rho)$. This appendix presents the catalytic extensions within the same geometric framework. We write $F_\beta(\rho)=\tr(\rho H)-\beta^{-1}S(\rho)$ for the nonequilibrium free energy and $C(\rho)$ for the set of resonant coherent modes. As shown in~\cite{shiraishi2025}, convertibility under correlated-catalytic thermal operations is governed by the free-energy order together with the mode-inclusion condition $C(\rho')\subseteq C(\rho)$.

We define the catalytic analogue of the magic future cone as
\begin{equation} \label{Eq:cat-cone}
  \mathcal{T}^{\mathcal{M}}_{+,\mathrm{cat}}(\rho):=
  \Bigl\{\sigma\notin\stab:
  F_\beta(\sigma)\le F_\beta(\rho)\ \text{and}\ C(\sigma)\subseteq C(\rho)\Bigr\}.
\end{equation}
When $C(\rho)\neq\{0\}$, we have the strict inclusion $\mathcal{T}^{\mathcal{M}}_{+}(\rho)\subsetneq\mathcal{T}^{\mathcal{M}}_{+,\mathrm{cat}}(\rho)$. This stems from the fact that the catalytic setting removes constraints such as the coherence cap $r\le\v{c}_{\max}(q)$ present in the noncatalytic construction. Consequently, the associated volume strictly increases:
\begin{equation}
  \mathcal{V}^{\mathcal{M}}_{+,\mathrm{cat}}(\rho)>\mathcal{V}^{\mathcal{M}}_{+}(\rho).
\end{equation}

In analogy with Eq.~\eqref{Eq:critical-temperature}, we define the catalytic critical inverse temperature by
\begin{equation}
  \label{eq:betacat}
  \beta_{\mathrm{crt}}^{\mathrm{cat}}(\rho,H)
  :=\inf\Bigl\{\beta\ge0:\ \exists\ \sigma\notin\stab\ \text{with}\ 
  F_\beta(\sigma)\le F_\beta(\rho),\ C(\sigma)\subseteq C(\rho)\Bigr\}.
\end{equation}
Since correlated catalysts extend the set of reachable states, $\beta_{\mathrm{crt}}^{\mathrm{cat}}(\rho,H)\le\beta_{\mathrm{crt}}(\rho,H)$, with equality only when the input state $\rho$ is energy-incoherent. For target states $\sigma$ diagonal in the energy basis of $H$, \eqref{eq:betacat} reduces to a one-variable root in the classical order $F_\beta(\sigma_z)\le F_\beta(\rho)$. Finally, in the catalytic regime, the population–coherence ``budget'' $\bigl[I_\beta(p),\v{c}_{\max}(q)\bigr]$ is superseded by a single free-energy inequality together with $C(\sigma)\subseteq C(\rho)$. Thus, all temperature thresholds are determined solely by free-energy comparisons. 

\end{document}